\documentclass[11pt]{amsart}

\usepackage[T1]{fontenc}
\usepackage{lmodern}
\usepackage{microtype}
\usepackage{amsmath,amssymb,mathtools,bm}
\usepackage[hidelinks]{hyperref}
\usepackage[nameinlink,noabbrev]{cleveref}
\usepackage{enumitem}
\usepackage{booktabs}
\usepackage{geometry}
\newtheorem{theorem}{Theorem}[section]
\newtheorem{proposition}[theorem]{Proposition}
\newtheorem{lemma}[theorem]{Lemma}
\newtheorem{corollary}[theorem]{Corollary}
\theoremstyle{definition}

\theoremstyle{remark}
\newtheorem{remark}[theorem]{Remark}

\DeclareMathOperator{\Tr}{Tr}
\DeclareMathOperator{\diag}{diag}
\newcommand{\Md}{M_d(\mathbb C)}
\newcommand{\Dplus}{\mathcal D_d^{+}}
\newcommand{\Htr}{\mathcal H_d^{0}}
\newcommand{\Id}{\mathbf 1}
\newcommand{\HS}{\mathrm{HS}}
\newcommand{\Riem}{\mathrm{Riem}}
\newcommand{\PhiK}{\Phi_K}
\newcommand{\ip}[2]{\left\langle #1,#2\right\rangle_{\HS}}

\title[A qutrit counterexample to Conjecture 4.9]{A minimal qutrit counterexample to Conjecture~4.9 of Lesniewski and Ruskai}
\author{Domingos S. P. Salazar}
\address{Unidade de Educa\c{c}\~ao a Dist\^ancia e Tecnologia, Universidade Federal Rural de Pernambuco, 52171-900 Recife, Pernambuco, Brazil}
\date{August 15, 2026}

\begin{document}

\begin{abstract}
Lesniewski and Ruskai conjectured that the contraction coefficient of every monotone Riemannian metric under a unital stochastic map equals the Hilbert--Schmidt contraction on the traceless subspace. We disprove the conjecture with an explicit entanglement-breaking qutrit channel induced by a doubly stochastic $3\times3$ matrix. A faithful diagonal state and a commuting traceless tangent give, simultaneously for every normalized monotone metric, the exact lower bound
\[
  \eta^{\Riem}_{\kappa}(\PhiK)\ge \frac{8896}{20007}
  >\frac{62+2\sqrt{61}}{225}
  =\Lambda_2(\PhiK^{\dagger}\PhiK).
\]
The counterexample is entirely classical on a maximal abelian subalgebra. A theorem of Hiai and Ruskai establishes the conjectured identity for all unital qubit maps, so dimension three is minimal among full matrix algebras.
\end{abstract}

\maketitle

\begin{center}
\begin{minipage}{0.94\textwidth}
\small\textbf{AI-use disclosure.}
GPT 5.6 Sol was used to assist with literature retrieval, exact symbolic arithmetic checks, and preparation of an initial draft. Responsibility for the mathematical statements and the final manuscript rests with the author.
\end{minipage}
\end{center}

\section{Introduction}

A monotone quantum Riemannian metric gives each faithful state $\rho$ a positive quadratic form on traceless Hermitian perturbations and contracts under stochastic maps. Lesniewski and Ruskai introduced contraction coefficients for this full class of metrics and, in Conjecture~4.9 of \cite{LR1999}, proposed a striking rigidity statement: if a stochastic map $\Phi$ is unital, then the global metric contraction should be independent of the chosen monotone metric and should equal the squared Hilbert--Schmidt norm of $\Phi$ on traceless matrices.

The conjectured value is necessarily a lower bound. At the maximally mixed state every normalized monotone metric is proportional to the Hilbert--Schmidt form, so the local contraction there is
\[
 \Lambda_2(\Phi^\dagger\Phi)
 =\sup_{\substack{A=A^*,\ \Tr A=0\\A\ne0}}
 \frac{\Tr|\Phi(A)|^2}{\Tr|A|^2}.
\]
The conjecture asserts that optimizing also over the base state can never increase this value.

We show that this assertion fails already for a classical doubly stochastic channel on three symbols, embedded as an entanglement-breaking channel on $M_3(\mathbb C)$. The same commuting witness works for every monotone quantum metric because all such metrics restrict to Fisher--Rao geometry on a commutative algebra. The exact gap is
\[
 \frac{8896}{20007}
 -\frac{62+2\sqrt{61}}{225}
 \approx 9.9664\times10^{-2}.
\]
The construction also identifies the obstruction: unitality fixes the uniform distribution, but the global Fisher contraction may be larger at a nonuniform distribution.

Hiai and Ruskai later proved the conjectured identity for every unital qubit map \cite{HR2016}. Consequently, the qutrit example below is dimension-minimal among channels $M_d\to M_d$.

\paragraph{Literature status.}
The conjecture was stated in the 1998 preprint and 1999 journal article \cite{LR1999}. A targeted search through August 15, 2026 for the exact conjecture number, its displayed equality, and forward citations did not locate a published resolution in dimension three or higher. Because the present construction is classical on the diagonal algebra, an equivalent counterexample could nevertheless occur in the classical $\chi^2$ strong-data-processing literature under different terminology. We therefore make no unconditional first-priority claim.

\section{Monotone metrics and the conjecture}

Let $\Md$ be the algebra of complex $d\times d$ matrices, equipped with the Hilbert--Schmidt inner product
\[
 \ip{X}{Y}=\Tr X^*Y.
\]
Let $\Dplus$ be the faithful density matrices and let
\[
 \Htr=\{A\in\Md:A=A^*,\ \Tr A=0\}
\]
be their common tangent space. For $\rho\in\Dplus$, denote left and right multiplication by
\[
 L_\rho(X)=\rho X,
 \qquad
 R_\rho(X)=X\rho.
\]

We use the standard normalized parametrization of monotone metrics \cite{Petz1996,PetzSudar1996,HR2016}. Let $\mathcal K$ be the set of functions $\kappa:(0,\infty)\to(0,\infty)$ that are operator convex and satisfy
\[
 \kappa(1)=1,
 \qquad
 x\kappa(x)=\kappa(x^{-1}).
\]
For $\kappa\in\mathcal K$, set
\begin{equation}\label{eq:omega}
 \Omega_\rho^\kappa
 =R_\rho^{-1}\,\kappa(L_\rho R_\rho^{-1}),
\end{equation}
and define
\begin{equation}\label{eq:metric}
 \gamma_\rho^\kappa(A,B)
 =\ip{A}{\Omega_\rho^\kappa(B)}.
\end{equation}
Every normalized monotone Riemannian metric is obtained in this way. The original $g$-parametrization of Lesniewski and Ruskai determines the same metric after symmetrization and normalization; positive rescaling does not affect a contraction coefficient.

For a completely positive trace-preserving map $\Phi:M_d\to M_{d'}$, its metric contraction coefficient is
\begin{equation}\label{eq:eta}
 \eta_\kappa^{\Riem}(\Phi)
 =\sup_{\rho\in\Dplus}
  \sup_{\substack{A\in\Htr\\A\ne0}}
 \frac{\gamma_{\Phi(\rho)}^\kappa(\Phi(A),\Phi(A))}
      {\gamma_\rho^\kappa(A,A)},
\end{equation}
where the supremum is restricted to $\rho$ for which $\Phi(\rho)$ is faithful. This restriction is immaterial for the witness below.

\begin{lemma}[Commuting reduction]\label{lem:commuting}
If $\rho\in\Dplus$ and $A\in\Htr$ commute, then for every $\kappa\in\mathcal K$,
\begin{equation}\label{eq:fisher}
 \gamma_\rho^\kappa(A,A)=\Tr \rho^{-1}A^2.
\end{equation}
Thus all normalized monotone metrics coincide with the Fisher--Rao metric on a commutative algebra.
\end{lemma}

\begin{proof}
Since $[\rho,A]=0$, one has
\[
 (L_\rho R_\rho^{-1})(A)=A.
\]
Functional calculus and $\kappa(1)=1$ therefore give
\[
 \Omega_\rho^\kappa(A)
 =R_\rho^{-1}\kappa(L_\rho R_\rho^{-1})(A)
 =A\rho^{-1}.
\]
Substitution into \eqref{eq:metric} yields \eqref{eq:fisher}.
\end{proof}

For a unital trace-preserving map $\Phi:M_d\to M_d$, let $\Phi^\dagger$ denote its Hilbert--Schmidt adjoint and set
\begin{equation}\label{eq:lambda2def}
 \Lambda_2(\Phi^\dagger\Phi)
 =\sup_{\substack{A=A^*,\ \Tr A=0\\A\ne0}}
 \frac{\Tr|\Phi(A)|^2}{\Tr|A|^2}.
\end{equation}
Lesniewski--Ruskai Conjecture~4.9 states that
\begin{equation}\label{eq:conjecture}
 \eta_\kappa^{\Riem}(\Phi)
 =\Lambda_2(\Phi^\dagger\Phi)
 \qquad\text{for every }\kappa\in\mathcal K
\end{equation}
whenever $\Phi$ is unital.

\section{The qutrit channel}

Consider the matrix
\begin{equation}\label{eq:K}
 K=
 \begin{pmatrix}
  \frac13&0&\frac23\\[1mm]
  \frac35&\frac25&0\\[1mm]
  \frac1{15}&\frac35&\frac13
 \end{pmatrix}.
\end{equation}
Every row and every column sums to one. Define $\PhiK:M_3\to M_3$ by
\begin{equation}\label{eq:channel}
 \PhiK(X)
 =\sum_{i,j=1}^3
 K_{ij}\langle j|X|j\rangle\,|i\rangle\langle i|.
\end{equation}

\begin{proposition}\label{prop:channel}
The map $\PhiK$ is completely positive, trace preserving, unital, and entanglement breaking.
\end{proposition}

\begin{proof}
The rank-one Kraus operators
\[
 V_{ij}=\sqrt{K_{ij}}\,|i\rangle\langle j|
\]
realize \eqref{eq:channel}. Column stochasticity gives
\[
 \sum_{i,j}V_{ij}^*V_{ij}=\Id,
\]
so $\PhiK$ is trace preserving. Row stochasticity gives
\[
 \sum_{i,j}V_{ij}V_{ij}^*=\Id,
\]
so it is unital. Since all Kraus operators have rank one---equivalently, the channel measures in the computational basis and prepares a diagonal state---it is entanglement breaking.
\end{proof}

\begin{proposition}[The conjectured spectral value]\label{prop:spectral}
For the channel \eqref{eq:channel},
\begin{equation}\label{eq:spectralvalue}
 \Lambda_2(\PhiK^\dagger\PhiK)
 =\frac{62+2\sqrt{61}}{225}.
\end{equation}
\end{proposition}

\begin{proof}
The channel annihilates every off-diagonal matrix unit. On diagonal matrices, identified with column vectors in $\mathbb R^3$, it acts as $K$. Consequently, $\PhiK^\dagger\PhiK$ acts as $K^{\mathsf T}K$ on the diagonal subspace and as zero on the off-diagonal subspace. Direct multiplication gives
\begin{equation}\label{eq:KtK}
 K^{\mathsf T}K=
 \begin{pmatrix}
  \frac{107}{225}&\frac7{25}&\frac{11}{45}\\[1mm]
  \frac7{25}&\frac{13}{25}&\frac15\\[1mm]
  \frac{11}{45}&\frac15&\frac59
 \end{pmatrix}.
\end{equation}
Its characteristic polynomial is
\begin{equation}\label{eq:charpoly}
 \det(\lambda\Id-K^{\mathsf T}K)
 =\frac1{225}(\lambda-1)
   (225\lambda^2-124\lambda+16).
\end{equation}
Because $K$ is doubly stochastic, the eigenvalue $1$ belongs to the identity vector $(1,1,1)^{\mathsf T}$. The other two diagonal eigenvalues, both lying in the traceless subspace, are
\[
 \lambda_\pm=\frac{62\pm2\sqrt{61}}{225}.
\]
The off-diagonal eigenvalues are zero. Hence the largest eigenvalue on the full traceless subspace is $\lambda_+$.
\end{proof}

\section{Universal violation of the conjecture}

\begin{theorem}[Counterexample to Conjecture~4.9]\label{thm:main}
For the unital entanglement-breaking qutrit channel $\PhiK$ in \eqref{eq:channel} and every normalized monotone Riemannian metric $\kappa\in\mathcal K$,
\begin{equation}\label{eq:mainineq}
 \eta_\kappa^{\Riem}(\PhiK)
 \ge \frac{8896}{20007}
 >\frac{62+2\sqrt{61}}{225}
 =\Lambda_2(\PhiK^\dagger\PhiK).
\end{equation}
In particular, Lesniewski--Ruskai Conjecture~4.9 is false.
\end{theorem}

\begin{proof}
Choose the faithful diagonal state and traceless diagonal tangent
\begin{equation}\label{eq:witness}
 \rho_\star
 =\diag\!\left(\frac1{10},\frac45,\frac1{10}\right),
 \qquad
 A_\star=\diag(0,1,-1).
\end{equation}
Both the input pair and its image under $\PhiK$ commute. By \cref{lem:commuting}, the metric quotient is therefore independent of $\kappa$.

At the input,
\begin{equation}\label{eq:inputmetric}
 \gamma_{\rho_\star}^\kappa(A_\star,A_\star)
 =\frac{1}{4/5}+\frac{1}{1/10}
 =\frac{45}{4}.
\end{equation}
On diagonal entries, the channel acts as $K$, and
\begin{equation}\label{eq:images}
 K
 \begin{pmatrix}1/10\\[1mm]4/5\\[1mm]1/10\end{pmatrix}
 =\begin{pmatrix}1/10\\[1mm]19/50\\[1mm]13/25\end{pmatrix},
 \qquad
 K
 \begin{pmatrix}0\\[1mm]1\\[1mm]-1\end{pmatrix}
 =\begin{pmatrix}-2/3\\[1mm]2/5\\[1mm]4/15\end{pmatrix}.
\end{equation}
Thus
\begin{align}
 \gamma_{\PhiK(\rho_\star)}^\kappa
 \bigl(\PhiK(A_\star),\PhiK(A_\star)\bigr)
 &=\frac{(2/3)^2}{1/10}
   +\frac{(2/5)^2}{19/50}
   +\frac{(4/15)^2}{13/25} \notag\\
 &=\frac{11120}{2223}.
 \label{eq:outputmetric}
\end{align}
The corresponding admissible quotient in \eqref{eq:eta} is
\begin{equation}\label{eq:quotient}
 \frac{11120/2223}{45/4}
 =\frac{8896}{20007}.
\end{equation}
It remains only to verify the strict comparison with \eqref{eq:spectralvalue}. We use the rational intermediate value $26/75$:
\begin{equation}\label{eq:firstcompare}
 \frac{8896}{20007}-\frac{26}{75}
 =\frac{49006}{500175}>0,
\end{equation}
and
\begin{equation}\label{eq:secondcompare}
 \frac{26}{75}-\frac{62+2\sqrt{61}}{225}
 =\frac{16-2\sqrt{61}}{225}>0,
\end{equation}
because $\sqrt{61}<8$. Combining \eqref{eq:quotient}--\eqref{eq:secondcompare} proves \eqref{eq:mainineq} for every $\kappa\in\mathcal K$.
\end{proof}

The proof can be stated without quantum notation. For a doubly stochastic matrix $K$, define its global Fisher--Rao contraction by
\begin{equation}\label{eq:classicaleta}
 \eta_{\mathrm{FR}}(K)
 =\sup_{\substack{p_i>0,\ \sum_i p_i=1}}
  \sup_{\substack{a\ne0,\ \sum_i a_i=0}}
 \frac{\displaystyle\sum_i\frac{(Ka)_i^2}{(Kp)_i}}
      {\displaystyle\sum_j\frac{a_j^2}{p_j}}.
\end{equation}
The uniform base point yields the squared second singular value of $K$. The witness in \eqref{eq:witness} proves that the global coefficient can be strictly larger.

\begin{corollary}[Classical form]\label{cor:classical}
For the doubly stochastic matrix $K$ in \eqref{eq:K},
\[
 \eta_{\mathrm{FR}}(K)
 \ge \frac{8896}{20007}
 >s_2(K)^2
 =\frac{62+2\sqrt{61}}{225},
\]
where $s_2(K)$ is the second singular value of $K$.
\end{corollary}

\section{Minimality and mechanism}

\begin{corollary}[Dimension minimality]\label{cor:minimal}
Among channels on full matrix algebras $M_d(\mathbb C)$, dimension $d=3$ is the smallest dimension in which Conjecture~4.9 can fail.
\end{corollary}

\begin{proof}
The one-dimensional case is trivial. Hiai and Ruskai proved that every unital positive trace-preserving qubit map, represented on the Bloch ball by a real matrix $T$, satisfies
\[
 \eta_\kappa^{\Riem}(\Phi_T)=\|T\|_\infty^2
 \qquad\text{for every }\kappa\in\mathcal K
\]
\cite[Theorem~6.1]{HR2016}. This is exactly the Hilbert--Schmidt spectral coefficient on the traceless qubit subspace. \Cref{thm:main} supplies a failure for $d=3$.
\end{proof}

\begin{remark}[Why unitality is insufficient]
At the maximally mixed state $\rho=\Id/d$, every normalized monotone metric is a scalar multiple of the Hilbert--Schmidt metric, so its local contraction is governed by $\Lambda_2(\Phi^\dagger\Phi)$. The true coefficient \eqref{eq:eta}, however, is a supremum over all base states. In the commutative sector, changing the base point changes the Euclidean norm into the inverse-covariance norm
\[
 a\longmapsto \sum_i\frac{a_i^2}{p_i}.
\]
The nonuniform state $p=(1/10,4/5,1/10)$ exposes a direction for which this reweighting is amplified by $K$ more strongly than the uniform-state singular value predicts. Thus the failure requires neither quantum coherence nor entanglement; it is already a failure of global Fisher geometry for a doubly stochastic Markov kernel.
\end{remark}

\begin{remark}[Uniformity over all monotone metrics]
The conclusion is stronger than finding one metric for which \eqref{eq:conjecture} fails. Since the witness and its image commute, the quotient \eqref{eq:quotient} is common to every normalized monotone metric. Hence one and the same state--tangent pair violates the conjectured value simultaneously for the entire Petz--Lesniewski--Ruskai family.
\end{remark}

\section*{Acknowledgments}
The author thanks the developers and maintainers of the open scientific literature infrastructure used to verify the historical statements and references.

\end{document}